\documentclass[review=false]{jfp-epi}

\setcopyright{none}
\usepackage{amsmath}
\usepackage{amssymb}
\usepackage{stmaryrd}

\usepackage{accents}

\usepackage{mathpartir}
\usepackage{mathtools}

\usepackage{enumerate}
\usepackage[inline]{enumitem}
\setlist[itemize]{noitemsep, topsep=6pt}
\newlist{inparenumA}{enumerate*}{1}
\setlist[inparenumA]{label=(\Alph*)}
\newlist{enumN}{enumerate}{1}
\setlist[enumN]{label=(\arabic*), noitemsep, topsep=6pt}
\newlist{enumNsep}{enumerate}{1}
\setlist[enumNsep]{label=(\arabic*), topsep=6pt}
\newlist{enumA}{enumerate}{1}
\setlist[enumA]{label=(\Alph*), noitemsep, topsep=6pt}
\newlist{enuma}{enumerate}{1}
\setlist[enuma]{label=(\alph*), noitemsep, topsep=6pt}

\usepackage{mathpartir}
\usepackage{fancyvrb}

\usepackage{array}  

\usepackage{xcolor}
\usepackage{graphicx}
\usepackage{tcolorbox}

\usepackage{mathtools}

\newcommand{\mytagform}[1]{\textsc{#1}}
\newtagform{labeq}[\mytagform]{(}{)}
\usetagform{labeq}

\newcommand{\nt}[1]{{\normalfont\textit{#1}}}
\newcommand{\OR}{\ \ |\ \ }
\newcommand\set[1]{\ensuremath{\{#1\}}}
\newcommand\rng[1]{\ensuremath{\textrm{rng}(#1)}}

\newcommand\IN{&\in&}
\newcommand\IS{&::=&}

\newcommand\Con{\textit{Con}}

\newcommand\Var{\textit{Var}}
\newcommand\EVar{\textit{EVar}}
\newcommand\IVar{\textit{IVar}}
\newcommand\BinOp{\textit{BinOp}}

\newcommand\Pat{\textit{Pat}}
\newcommand\EPat{\textit{EPat}}
\newcommand\Exp{\textit{Exp}}

\newcommand\Seg{\textit{Seg}}

\newcommand\EExp{\textit{EExp}}
\newcommand\BExp{\textit{BExp}}

\newcommand\Subst{\textit{Subst}}

\newcommand\Core{\nt{LC}}
\newcommand\cexp{\ensuremath{\ell}}
\newcommand\gplus{\ensuremath{\oplus}}
\newcommand\subs[1]{[b/a]#1}

\newcommand\op{\ensuremath{\bullet}}
\newcommand{\expr}{\nt{e}}
\newcommand{\ellibar}{\raisebox{-1pt}{\kern1pt$\scriptscriptstyle\ldots$}}
\newcommand{\ent}[1]{\ensuremath{\overset\ellibar{#1}}}
\newcommand{\eexpr}{\ent{\nt{e}}}
\newcommand{\eexp}{\ent{\nt{e}}}
\newcommand{\bexp}{\ensuremath{\bar{e}}}
\newcommand{\bexpa}{\ensuremath{\bexp_l}}
\newcommand{\bexpb}{\ensuremath{\bexp_r}}

\newcommand{\evar}[1][x]{\ent{\nt{#1}}}
\newcommand{\ivar}[1][\iota]{\ent{#1}}
\newcommand{\epat}{\ent{\nt{p}}}

\newcommand{\seg}{\nt{s}}
\newcommand\metadots{\cdots}
\newcommand{\sequ}[2]{\ensuremath{#1#2\metadots#2#1}}

\newcommand{\sequx}[3]{\ensuremath{#1#2\metadots#2#3}}

\newcommand{\win}[1]{#1}

\newcommand{\abs}[2]{\ensuremath{\BS#1\prog{.}#2}}
\newcommand{\xlet}[3]{\ensuremath{\prog{let}\,#1\prog{=}#2\,\prog{in}\,#3}}
\newcommand{\xrule}[2]{\ensuremath{#1\,\prog{->}\,#2}}
\newcommand{\case}[2]{\ensuremath{\prog{case}\,#1\,\prog{of}\,\BL#2\BR}}

\newcommand\pI{\prog{1}}

\newcommand\e[1]{\ensuremath{\expr_{#1}}}

\newcommand{\subst}{\ensuremath{\theta}}
\newcommand{\sub}[3]{\ensuremath{[#1/#2]#3}}
\newcommand{\emptyEnv}{\ensuremath{\varnothing}}

\newcommand{\lCase}{\textsc{Case}}

\newcommand\lft[1]{\ensuremath{\dot{#1}}}
\newcommand\rgt[1]{\ensuremath{\ddot{#1}}}
\renewcommand\lft[1]{\ensuremath{#1_1}}
\renewcommand\rgt[1]{\ensuremath{#1_2}}

\newcommand\eL{\lft{e}}
\newcommand\eR{\rgt{e}}
\newcommand\beL{\lft{\bexp}}
\newcommand\beR{\rgt{\bexp}}

\newcommand\xL{\lft{x}}
\newcommand\xR{\rgt{x}}

\newcommand\gen[5][\subst]{\ensuremath{#1\vdash#2|#3\Rightarrow#4,#5}}

\def\trSegSym{{%
    \setbox0\hbox{\ensuremath{\,\to\,}}%
    \rlap{\hbox to \wd0{\hss\raisebox{3pt}{\ensuremath{
    \scriptstyle\cdots}}\,\hss}}\box0
}}

\def\trCaseRuleSym{{%
    \setbox0\hbox{\ensuremath{\,\longrightarrow\,}}%
    \rlap{\hbox to \wd0{\hss\raisebox{3pt}{\ensuremath{
    \scriptstyle\to}}\,\hss}}\box0
}}

\newcommand{\trE}[2]{\ensuremath{#1\to#2}}
\newcommand{\trSeg}[2]{\ensuremath{#1\trSegSym#2}}
\newcommand{\trCaseRule}[2]{\ensuremath{#1\trCaseRuleSym#2}}
\newcommand{\sliceOrRange}[2]{\ensuremath{\|#1,#2\|}}

\newcommand\Lft{\textit{left}}
\newcommand\Mdl{\textit{middle}}
\newcommand\Rgt{\textit{right}}
\newcommand\Loc{\textit{Loc}}
\newcommand\sem[1]{\ensuremath{\llbracket#1\rrbracket}}
\newcommand\rangeName{\textit{isRange}}
\newcommand{\isRange}[2]{\rangeName(#1,#2)}

\newcommand\cmp[2]{\ensuremath{#1{\scriptstyle\gtreqless}#2}}
\newcommand\traversalName{\textit{Traversal}}
\newcommand\traversal[2]{\ensuremath{\traversalName(#1,#2)}}

\usepackage{fancyvrb}

\definecolor{myblue}{rgb}{0.1, 0.1, 0.55}

\usepackage{relsize}
\newcommand{\progfontsize}{\relsize{-0.5}} 
\newcommand{\progcolor}{myblue}
\newcommand{\progcolorD}{myblue}
\newcommand{\progind}{\parindent}
\newcommand{\prog}[1]{\textrm{\progfontsize\texttt{\color{\progcolor}#1}}}
\def\progColumn{\progfontsize\ttfamily\color{\progcolor}}

\newlength{\mylen}
\newcommand\fwdsp{\kern2.2pt}
\newcommand\bcksp{\kern-1.9pt}
\def\edot{\makebox[0pt][c]{\prog{.}}}
\newcommand\E{\makebox[2\mylen][c]{\edot\hfil\edot\hfil\edot}}

\newcommand\BS{\prog{\symbol{92}}}  
\newcommand\BL{\prog{\symbol{123}}}  
\newcommand\BR{\prog{\symbol{125}}}  

\newcommand\V{\raisebox{-2pt}{\rotatebox{90}{\E}}}

\makeatletter
\let\ttfamily\relax 
\DeclareRobustCommand\ttfamily
        {\not@math@alphabet\ttfamily\mathtt
         \fontfamily\ttdefault\selectfont\frenchspacing}
\makeatother

\fvset{listparameters=\setlength{\topsep}{4pt}\setlength{\partopsep}{4pt}}

\DefineVerbatimEnvironment{program}{Verbatim}
  {baselinestretch=1.0,xleftmargin=\progind,fontsize=\progfontsize,
   formatcom=\color{\progcolorD},
   samepage=true}

\DefineVerbatimEnvironment{programB}{BVerbatim}
  {baselinestretch=1.0,xleftmargin=\progind,fontsize=\progfontsize,
   formatcom=\color{\progcolorD},
   samepage=true,baseline=t,
   commandchars=\\\{\}}

\DefineVerbatimEnvironment{programC}{BVerbatim}
  {baselinestretch=1.0,fontsize=\progfontsize,
   formatcom=\color{\progcolorD},
   samepage=true,baseline=t,
   commandchars=\\\{\}}

\DefineVerbatimEnvironment{programCmd}{Verbatim}
  {baselinestretch=1.0,xleftmargin=\progind,fontsize=\progfontsize,
   formatcom=\color{\progcolorD},
   samepage=true,
   commandchars=\\\{\}}

\def\colsepFill{\hfill\vline\hfill}

\newenvironment{derivation}
  {\par\vskip\abovedisplayskip 
   \begin{tabular}{@{}c@{\ \ }>{\progColumn}l@{\qquad}l}}
  {\end{tabular}\vskip\belowdisplayskip\ignorespacesafterend}

\usepackage{noindentafter}
\NoIndentAfterEnv{equation}
\NoIndentAfterEnv{align}
\NoIndentAfterEnv{derivation}
\NoIndentAfterEnv{program}
\NoIndentAfterEnv{programCmd}
\NoIndentAfterEnv{enuma}

\begin{document}

\title{Programming with Ellipses}

\author{Martin Erwig}
\orcid{0000-0002-7471-4554}
\affiliation{%
  \institution{Oregon State University}
  \country{USA}
  \authoremail{erwig@oregonstate.edu}}
\author{Hangil Kim}
\orcid{0009-0004-2357-9016}
\affiliation{%
  \institution{Oregon State University}
  \country{USA}
  \authoremail{kimhang@oregonstate.edu}
}

\begin{abstract}
We present the design of a language extension that allows the use of ellipses (``...'') in patterns and expressions, which facilitates function definitions on lists that are more succinct and direct than the standard recursive ones. 
The semantics of ellipsis notation is defined via a program translation that is based on computing least general generalizations of the expressions on the boundary of ellipses.
We show that ellipsis notation applies to a wide spectrum of function definitions. 
Ellipsis notation can also support the teaching of functional programming, and its inherently iterative nature may make it especially helpful to those with an imperative programming background.
%
Like list comprehensions, ellipsis notation provides an attractive tool for functional languages that adds to their syntactic variety and enriches their expressiveness.
\end{abstract}

\maketitle

\section{Introduction}


Lists are probably the most widely used data structures in functional programming. 
It is therefore no surprise that functional languages offer special support for programming with lists. In addition to providing comprehensive libraries, such support can be observed in the language design, especially in its syntax. 

Many functional languages, for instance, offer special syntax for list constants.
%
In Haskell \cite{Haskell98}, one such syntax is a simple form of ellipsis notation.
For example, the expression \prog{[a..b]} constructs the list of elements between \prog{a} and \prog{b}, which works for elements of any \prog{Enum} type. 
One can even specify a different offset, as in \prog{[1,3..9]}, which yields the list of odd digits.
In addition, Haskell offers \emph{list comprehensions} \cite{Wad85} as a special syntax for constructing and filtering lists. List comprehensions are a simple, yet effective language extension whose popularity has led to their adoption by other programming languages.
The success of the list comprehension notation is in part due to its origin in elementary set theory, where it is known as \emph{set comprehension} or \emph{set-builder} notation \cite{DiscreteMath}. 
The close relationship between list and set comprehensions means that programmers with a background in basic discrete math tend to find the notation familiar and immediately accessible.

As shown in the Haskell report \cite{Haskell98}, list comprehensions are ultimately syntactic sugar for nested applications of \prog{concatMap} plus conditionals, but can often be rewritten more simply using \prog{map} and \prog{filter}.
Arguably, list comprehension syntax helps in writing simpler Haskell programs.
Despite their seemingly limited contribution, list comprehensions are quite popular and must be considered a successful language extension.
A similar story can be told about \emph{do notation}, which was introduced to Haskell to simplify cascades of monadic bind operations and later adopted by several other programming languages.


Encouraged by the success of list comprehensions, we propose the use of \emph{ellipses} (three consecutive periods ``...'') for programming with lists. Ellipses are an even more widespread and simpler notation than set comprehensions, and thus promise also to be accessible by a large audience. 
As a simple example, consider the function \prog{map}, whose well-known recursive definition is as follows.

\begin{program}
map f []     = []
map f (x:xs) = f x:map f xs
\end{program}
This definition is as clear as it gets, but it requires the use of recursion and two equations to distinguish the base case and the inductive case. 
Now, consider a definition of \prog{map} that uses ellipses instead.
\label{def:map}

\begin{programCmd}
map f [x1, \E, xn] = [f x1, \E, f xn]
\end{programCmd}

This definition expresses the effect of applying \prog{f} to all elements of the list more directly than the recursive one, since it contains only applications of \prog{f} to list elements. By contrast, the recursive definition additionally applies \prog{map} to \prog{f} and \prog{map f} to \prog{xs}, and to understand what that means, one has to unfold the definition of \prog{map}.
The definition with ellipses owes its conciseness to the fact that the \textit{ellipsis expression} defining the result on the RHS can refer to the elements laid bare by the \textit{ellipsis pattern} on the LHS.
The notation has actually been used as an explanation for the recursive definition of \prog{map} \cite[p.\ 91]{Rea89}.

A slightly different form of ellipsis expression can describe the aggregation of list elements with a binary operation.
As an example, consider the mathematical definition of the factorial function $n!=1\times\ldots\times n$. 
%
Expressing this definition in Haskell looks as follows.

\begin{programCmd}
fac n = 1 * \E * n
\end{programCmd}

Adding ellipsis notation to a programming language can make programming with lists more effective by providing the following two benefits.

First, ellipses can help express programs in a simpler way.
At the same time, ellipses are widely known and can thus be expected to be well received and accepted, promising a huge potential for supporting and simplifying the programming with lists.

Second, ellipses can aid novice programmers in learning functional programming with lists. 
In particular, programmers who learned how to program with an imperative programming language often have difficulties understanding and using recursion \citep{McCauley2015}. 
Ellipses are inherently iterative and thus make programmers with an imperative programming mindset feel at home when starting to program in a functional language.

\subsection*{This Paper}

In Section \ref{sec:progelli} we will discuss the different forms of ellipsis patterns and expressions and how they can be employed to program with lists. We present additional opportunities for the applications of ellipsis notation in the derivation of recursive function definitions and program traces in Section \ref{sec:otherapplic}.
In Section \ref{sec:synsem} we define the syntax and semantics of (basic) ellipsis expressions. We then illustrate in Section \ref{sec:splitpatterns} how we can extend the expressiveness of the notation via a simple extension of ellipsis patterns called \emph{split patterns}.
%
%
We evaluate the potential impact of ellipsis notation in Section \ref{sec:eval} by assessing how many definitions in a representative collection of list functions can be simplified through its use, and we compare its scope with that of list comprehensions.
Finally, we discuss related work in Section \ref{sec:relaw} and present our conclusions in Section \ref{sec:concl}. 
Overall, our paper makes the following contributions.

\begin{itemize}
\item \emph{Language extension.} We define the syntax and semantics of a basic and extended version of ellipsis notation to support the programming with lists.

\item \emph{Potential teaching support.} We illustrate how ellipses can be used in the teaching of functional programming by (a) providing a basis for succinct, general program traces and (b) supporting the derivation of recursive function definitions from iterative ones.
%

\item \emph{Evaluation.} We evaluate the scope and expressiveness of ellipsis notation by measuring the impact it can have on simplifying existing function definitions, and we compare its scope with that of list comprehensions.

\end{itemize}

\section{Programming with Ellipses}
\label{sec:progelli}

The two examples in the introduction already illustrate the basic idea behind programming with ellipses. 
Specifically, the definition of \prog{map} illustrates how an intended transformation of a list works through an interplay of an \emph{ellipsis pattern} \prog{[x1, \E, xn]} that provides bindings for the elements of a list and an \emph{ellipsis expression} \prog{[f x1, \E, f xn]} that uses these bindings to describe the resulting list.
But, there remain several additional aspects of the notation's support for programming that we will discuss next.
In the following discussion we refer to the function definitions shown in Figure \ref{fig:elliex}.

\begin{figure}
\hfill
\begin{programC}
reverse [x1, \E, xn] = [xn, \E, x1]

head [x1, \E, xn] = x1
last [x1, \E, xn] = xn
tail [x1, \E, xn] = [x2, \E, xn]
init [x1, \E, xn] = [x1, \E, x\BL{}n-1\BR]

length    [x1, \E, xn] = n
enumerate [x1, \E, xn] = [(1,x1), \E, (n,xn)]

foldl  f u [x1, \E, xn] = u `f` x1 `f` \E `f` xn
foldl1 f   [x1, \E, xn] = x1 `f` \E `f` xn

append [x1, \E, xn] [y1, \E, ym] = [x1, \E, xn, y1, \E, ym]
zip    [x1, \E, xn] [y1, \E, ym] = [(x1,y1), \E, (xn,ym)]
\end{programC}
\hfill{\ }
\caption{Some basic ellipsis programming examples.}
\label{fig:elliex}
\end{figure}

First, the numbering that is part of the variable names representing list elements in an ellipsis pattern is meaningful and can be exploited in several ways. 

\begin{itemize}
\item \emph{Restructuring lists:}
%
The numbers \prog{1} and \prog{n} in the LHS pattern establishes an ordering of the list elements that can be exploited and manipulated on the RHS, as illustrated by \prog{reverse}. \citet[p.\ 33]{FH88} actually use the definition with ellipses as an explanation for the recursive definition of that function.

\item \emph{Extracting parts of lists:}
The numbering of list elements enables direct access to individual elements at (or close to) the beginning and end of a list, facilitating element selection or manipulation, as illustrated by the functions \prog{head}, \prog{last}, \prog{tail}, and \prog{init}.

\item \emph{Deriving list information:}
The numbering information itself can be extracted and combined with other data, as in
\prog{length} and \prog{enumerate}.
\end{itemize}
Second, the definition of \prog{fac} illustrates an instance of a more general \textit{fold expression} that allows the definition of \prog{foldl} (and \prog{foldl1}) functions.
Several textbooks 
use a similar notation to explain recursive function definitions, for example, 
\cite[pp.\ 66-67]{BW88}, 
\cite[p.\ 51]{FH88}, 
\cite[pp.\ 67-69]{Hud00}, 
\cite[pp.\ 66-67]{Hut07}, 
\cite[p.\ 5]{Pau96}, 
\cite[p.\ 96]{Rea89},
\cite[p.\ 222]{Tho11}, or
\cite[p.\ 80]{Ull98}.

%
%

Third, we are able to use ellipsis patterns for multiple arguments, which facilitates the definition of functions such as \prog{append} and \prog{zip} (again used as an explanation in 
\cite[p.\ 95]{Rea89}).
The definition of \prog{append} also illustrates that list expressions may contain multiple ellipsis expressions. In general, a list expression has the form \prog{[}$\seg_1$\prog{,}$\cdots$\prog{,}$\seg_n$\prog{]} where each $\seg_i$ (called a \emph{segment}) is either a non-ellipsis expression or an ellipsis expression.



Finally, we can use equations using ellipses in conjunction with other equations. This can be particularly helpful to isolate special cases in function definitions.
Consider, for example, the definition of the function \prog{sum}. It is possible to give a single-line definition using ellipses as follows. 

\begin{programCmd}
sum [x1, \E, xn] = 0 + x1 + \E + xn
\end{programCmd}

However, here the initial ``\prog{0 +}'' term may seem surprising. Yet, the unit value for the sum is required. Without it, the function definition would be partial and cause a runtime error whenever \prog{sum} is applied to the empty list. Since this aspect might not be obvious, the following definition with two equations that separates the special case for the empty list might seem preferable.

\begin{programCmd}
sum []           = 0
sum [x1, \E, xn] = x1 + \E + xn
\end{programCmd}

Separating the case for the unit element for \prog{sum} (and similarly for other list aggregations) helps keep the ellipsis notation focused on what it can do best, namely, relating the list patterns of arguments to those of results. 

The meaning of the shown function definitions should be obvious to anyone who is familiar with ellipsis notation.\footnote{The definition of \prog{zip} is an exception because it raises the question of what the semantics of \prog{(xn,ym)} is. As we will explain in Section \ref{sec:semantics}, the semantics is the same as for the familiar \prog{zip} function.}
This is not surprising and is, in fact, a crucial motivation for the proposed language extension, namely, to leverage familiarity with specific linguistic structures for narrowing the gap between the cognitive intent of (novice) programmers and the artifacts they have to read and write.

The applicability of ellipsis notation extends beyond simple maps and folds and can express, more generally, folds over applications of $k$-ary zips for processing (parts of) multiple lists.
As an example consider the definition of the function \prog{sorted}, which determines whether the elements of a list is sorted by comparing adjacent list elements.

\begin{programCmd}
sorted [x1, \E, xn] = True && x1<=x2 && \E && x\BL{}n-1\BR<=xn
\end{programCmd}

As we will explain in Section \ref{sec:semantics}, this definition will be translated into a traversal over two sections of the list \prog{x} to produce a list of Booleans to be aggregated by a fold.
Specifically, the translation semantics produces the following function definition, which uses \prog{fwd} to create two forward-directed traversals of the list \prog{x} and the standard function \prog{zipWith} for zipping two lists with a binary function.

\begin{program}
sorted = \x.let n = length x in 
            foldl (&&) True (zipWith (\v1 v2.v1<=v2) 
                                     (fwd 1 (n-1) x) (fwd 2 n x))
\end{program}

\section{Ellipsis Notation in a Supporting Role}
\label{sec:otherapplic}

The introduction of ellipses is \emph{not} meant to oversimplify programming by getting novices hooked on a notation and limit the set of programming concepts they are exposed to. In addition to providing a smooth entry point, ellipses can also support the learning of recursion.
In this section we illustrate three ways in which this can work.

\subsection{Explaining Recursive Definitions}
\label{sec:elliexpl}

In addition to simplifying the definitions of functions, ellipsis notation can be leveraged to explain recursive function definitions.
Consider, for example, the following definition of the function \prog{(++)} for appending two lists.

\begin{program}
[]     ++ ys = ys
(x:xs) ++ ys = x:(xs ++ ys)
\end{program}

Programmers with experience in using recursion have no problem understanding this definition. However, to many novice functional programmers (who often have a background in imperative programming) recursion is a difficult concept \citep{McCauley2015}, and it is not obvious how this program works or that it even works at all.


How could we explain the definition?
Such an explanation must appeal to some model of lists that serves as an accepted foundation, and ellipsis notation could serve as one. Specifically, we expect \prog{(++)} to behave as follows.
%
\begin{equation}\tag{Append}\label{eq:append}
\prog{[x1, \E, xn] ++ [y1, \E, ym]} = \prog{[x1, \E, xn, y1, \E, ym]}
\end{equation}

%
Assuming that a programmer understands this notation, we can use it to show that the recursive definition satisfies that specification. We use the following equation (see also \cite[p.\ 148]{BW88} or \cite[p.\ 69]{Pau96}) for expanding ellipsis forms for non-empty lists (that is, $\prog{n}\geq 1$), which links the ellipsis list syntax to the usual cons notation employed in the recursive definition.
\begin{equation}\tag{Cons}\label{eq:cons}
\prog{[x1, \E, xn]} = \prog{[x1, x2, \E, xn]} = \prog{x1:[x2, \E, xn]}
\end{equation}

Then, we can illustrate how \prog{(++)} works for non-empty first arguments.

\begin{derivation}
  & [x1, \E, xn] ++ [y1, \E, ym] \\
= & (x1:[x2, \E, xn]) ++ [y1, \E, ym]  &  \eqref{eq:cons} \\
= & x1:([x2, \E, xn] ++ [y1, \E, ym])  &  (Definition of \prog{(++)}) \\
= & x1:[x2, \E, xn, y1, \E, ym]        &  (By induction) \\
= & [x1, \E, xn, y1, \E, ym]           &  \eqref{eq:cons}
\end{derivation}

Since one step involves induction, this derivation is not trivial for novices, but it is at least an explanation. Without it, the definition of \prog{(++)} can be illustrated through individual example traces.
The shown derivation is actually a summary of all traces from which arbitrary example traces can be derived by substituting concrete lists for the ellipsis expressions (and repeating the application of \eqref{eq:cons} and \prog{(++)} as needed).

As another example we can illustrate that the following recursive definition of \prog{reverse} is correct.

\begin{program}
reverse []     = []
reverse (x:xs) = reverse xs ++ [x]
\end{program}

The correctness criterion for \prog{reverse} can again be expressed using ellipsis notation, that is, by the definition of \prog{reverse} given in Figure \ref{fig:elliex}.
%
%
%
In addition to \eqref{eq:cons}, we need its dual \eqref{eq:snoc} for expanding ellipses at the end of lists (again, assuming $\prog{n}\geq 1$).
\begin{equation}\tag{Snoc}\label{eq:snoc}
\prog{[xn, \E, x1]} = \prog{[xn, \E, x2, x1]}
\end{equation}

We can then derive the effect of \prog{reverse} as follows.

\begin{derivation}
  & reverse [x1, \E, xn] \\
= & reverse (x1:[x2, \E, xn])     &  \eqref{eq:cons} \\
= & reverse [x2, \E, xn] ++ [x1]  &  (Definition of \prog{reverse}) \\
= & [xn, \E, x2] ++ [x1]          &  (By induction) \\
= & [xn, \E, x2, x1]              &  (Definition of \prog{(++)}) \\
= & [xn, \E, x1]                  &  \eqref{eq:snoc} 
\end{derivation}

%
%
%
% ...
%
%

%

\subsection{Deriving Recursive Definitions}
\label{sec:deriving}

A related application is the derivation of recursive function definitions from functions that are defined using ellipsis syntax.
As an example we consider again the definition of the function \prog{reverse}.

In the following transformation we add $L$ and $R$ tags to indicate whether the named transformation has been applied to the left or right side of the equation on the preceding line.

\newcommand{\LorR}[1][R]{\ \ensuremath{#1}}

\begin{derivation}
    & reverse [x1, \E, xn]~~~~~~= [xn, \E, x1] 
\\
$\Leftrightarrow$ & reverse (x1:[x2, \E, xn])  = [xn, \E, x1]  
    & \eqref{eq:cons} \LorR[L] \\
$\Leftrightarrow$ & reverse (x1:[x2, \E, xn])  = [xn, \E, x2, x1]  
    & \eqref{eq:snoc} \LorR \\
$\Leftrightarrow$ & reverse (x1:[x2, \E, xn]) = [xn, \E, x2] ++ [x1]  
    & \eqref{eq:append} \LorR \\
$\Leftrightarrow$ & reverse (x1:[x2, \E, xn]) = reverse [x2, \E, xn] ++ [x1]  
    & (Induction) \LorR \\
$\Leftrightarrow$ & reverse (x1:xs)           = reverse xs ++ [x1]  
   & (\E-to-var) \LorR[L,R] \\
\end{derivation}

\noindent
Note that the last step replaces an ellipsis pattern and expression by a variable.

\begin{figure}
\hfill
\begin{programC}
fac 1 = 1
fac n = n * fac (n-1)
\end{programC}
\hfill
\begin{programC}

fac n = n * \E * 1
\end{programC}
{\ \ \ }
\\[2ex]
\begin{programC}
  fac 6
= 6 * fac 5
= 6 * 5 * fac 4
  \fwdsp\V
= 6 * 5 * 4 * 3 * 2 * fac 1
= 6 * 5 * 4 * 3 * 2 * 1
= 720
\end{programC}
\colsepFill
\begin{programC}
  fac 6
= 6 * fac 5
= 6 * 5 * fac 4
  \fwdsp\V
= 6 * 5 * \E * fac 1
= 6 * 5 * \E * 1
= 720
\end{programC}
\colsepFill
\begin{programC}
  fac 6
= 6 * \E * 1
  \fwdsp\V
= 6 * 5 * \E * 1
= 720
\end{programC}
\hfill{}
\caption{Partial traces for computing \prog{fac 6}.
Left: A vertical ellipsis\ \ \V\ \ to omit steps from the derivation.
%
Middle: Language ellipses additionally simplify expressions.
%
Right: A trace generated by the non-recursive ellipsis-based definition. The use of different ellipsis symbols (horizontal vs.\ vertical) reflects the different roles of the ellipses.
}
\label{fig:trace}
\end{figure}

\subsection{Generating Concise Traces With Ellipses}
\label{sec:ellitrace}

Ellipses are also often used in abbreviating traces that explain the evaluation of functional programs \cite{Bir98,Tho11} where the ellipsis symbol is meant to denotes a sequence of omitted steps.

As an example, consider the recursive definition of the factorial function at the top of Figure \ref{fig:trace} and its ellipsis-based counterpart next to it.
Beneath the function definitions we show three versions of a trace for the computation of \prog{fac 6}.
The use of ellipses simplifies traces in two principal ways. First, we can replace intermediate steps that are repetitions of earlier steps just with different arguments by a vertical ellipsis (left trace).
Second, having ellipses available as lexical tokens allows the further simplification of such traces by eliding the middle parts of cascading applications of infix operations (middle trace).
Finally, by appealing to the ellipsis-based definition for \prog{fac} we can obtain an even shorter trace, since we avoid having to expand recursive function calls (right trace).


\section{Syntax and Semantics of Ellipsis Notation}
\label{sec:synsem}

In short, ellipsis notation provides a succinct syntax for defining functions that are potential folds over a \prog{zipWith} of $k$ lists, some of which may be reversed and stripped of prefixes and suffixes. 
In the following we define the translation semantics of ellipsis notation in three steps.
First, we describe the syntax of extending a small functional language by ellipsis patterns and expressions in Section \ref{sec:syntax}.
In Section \ref{sec:antiunif} we then explain the inference of functions that interpolate between the two boundary expressions of an ellipsis expression. These functions play a key role as the parameters for the \prog{zipWith} functions in the generation of the values denoted by an ellipses. 
The translation semantics described in Section \ref{sec:semantics} is based on this inference.

\subsection{Syntax}
\label{sec:syntax}

Consider the syntax of a typical core functional language, shown in Figure \ref{fig:syncore}. Ellipsis notation extends it by ellipsis patterns (\epat) and ellipsis expressions (\eexpr) as shown in Figure \ref{fig:syntax}. In the syntax definitions we make use of the following notation for extending syntactic domains. Given a definition such as $a \in A\ ::=\ R_1\ |\ \cdots\ |\ R_n$, we use the notation $b \in B ::= A\ \gplus\ S_1\ | \cdots\ |\ S_k$ to denote the definition $b \in B\ ::=\ \subs{R_1}\ |\ \cdots\ |\ \subs{R_n}\ |\ S_1\ |\ \cdots\ |\ S_k$ (where $\subs{R}$ denotes the substitution of all occurrences of metavariable $a$ in $R$ by metavariable $b$).

\begin{figure}[t]
\(
\begin{array}{r@{\ }c@{\ }l@{\ \ }c@{\ \ }l}
\multicolumn{5}{l}{
c \in \Con \qquad 
x \in \Var \qquad 
p \in \Pat\ ::=\ c\,\sequ{p}{} \OR x \OR \epat
}
\\
\cexp \IN \Core \IS c \OR x \OR \abs{x}{\cexp} \OR \cexp\;\cexp
\\
e \IN \Exp \IS \Core\ \gplus\ \xlet{x}{e}{e} \OR
               \case{e}{\sequ{\xrule{p}{e}}{\prog{;}}}
               \OR \eexpr 
\end{array}
\)
\caption{Syntax of the core functional language with anticipated extensions by ellipsis patterns (\epat) and ellipsis expressions (\eexpr).
}
\label{fig:syncore}
\end{figure}

Ellipsis patterns are written like list constants with an initial and final variable, separated by commas from an ellipsis symbol \E\ in the middle. 
The variables in an ellipsis pattern consist of two parts, a name \evar\ plus an index, which can be given either by the constant \prog{1} or an index variable \ivar.  
The two separate syntactic categories for ellipsis and index variables allows us to combine them into one syntactic entity by simple concatenation, avoiding the need for curly braces for separating single index variables.
In this paper we use \prog{x}, \prog{y}, and \prog{z} for \evar\ and \prog{k}, \prog{n}, and \prog{m} for \ivar, but in general, one can envision some form of ``ellipsis variable'' declaration (similar to infix declarations) that gives programmers control over the parsing process with the ability to choose the variable names to be used in ellipsis patterns and expressions.
However, we do stipulate that the names of both occurrences of \evar\ in an ellipsis pattern be identical, which means that patterns such as \prog{[x1,\E,yn]} are syntactically incorrect.

An ellipsis expression can take on one of three forms.
(1) The most basic form is given by two expressions connected by an ellipsis symbol, a so-called \emph{span}. Typical examples are the RHSs of the function definitions for \prog{reverse}, \prog{init}, or \prog{map}. 
But note that only a subset of expressions $\bexp\in\BExp$ (so-called \emph{boundary expressions}) can be used here; \BExp\ specifically excludes let and case expressions as well as references to ellipsis variables. 
%
%
As a slight generalization, a list expression consists of one or more so-called \emph{segments}, which can be ordinary expressions or spans. The RHS of \prog{append} shows an example that contains two spans.

(2) As a generalization of spans we also have \emph{fold expressions}, which insert a binary operator $\op\in\BinOp\subseteq\Con$ between the boundary expressions. A fold expression may also contain an additional expression and operator on the left side to define the value of the fold for an empty list.
The definitions of \prog{foldl}, \prog{sum}, and \prog{fac} serve as examples.

(3) The final form of an ellipsis expression is a reference to an ellipsis variable, as illustrated in the definitions of \prog{head} and \prog{last}. Note that the definition of \prog{length} does not fall under this case, since the RHS is given by a plain variable whose value is determined through the semantics of the core language. By contrast, ellipsis variable references such as \prog{x1} or \prog{xn} require a translation into an indexing operation on a list.

\begin{figure}[t]
\(
\begin{array}{r@{\ }c@{\ }l@{\ \ }c@{\ \ }l}
\epat \IN \EPat \IS \prog{[}\evar\pI\prog{,\E,}\evar\ivar\prog{]}
\qquad\evar,\evar[y]\in\EVar\qquad\ivar,\ivar[\kappa]\in\IVar \\
\eexp \IN \EExp \IS \prog{[}\sequ{\seg}{\prog{,}}\prog{]} 
        \OR   [e\ \op]\ \bexp\ \op\ \E\ \op\ \bexp
        \OR   \evar\BL\bexp\BR \\
\seg \IN \Seg \IS e \OR \bexp\prog{,\E,}\bexp \\
\bexp \IN \BExp \IS \Core\ \gplus\ \evar\BL\cexp\BR \\
\end{array}
\)
\caption{Syntax of ellipsis patterns and expressions. Note the difference between the ellipsis token ``\E'', which is part of the concrete syntax, and the ellipsis ``$\metadots$'', which is part of the meta-notation. 
An ellipsis pattern is written as a list constant with two numbered variables and an ellipsis symbol in the middle. 
The most basic form of an ellipsis expression is a \emph{span} $\bexp\prog{,\E,}\bexp$, which consist of two boundary expressions connected by commas and the ellipsis symbol.
In general an ellipsis expression is given by (1) a sequence of list segments \seg, 
each of which can either be an expression \expr\ or a span, (2) a fold expression, which is like a span with a binary operator \op\ between the boundary expressions (possibly extended by an expression defining the case for empty lists), or (3) a reference to an ellipsis variable $\evar\BL\bexp\BR$.
}
\label{fig:syntax}
\end{figure}

\subsection{Computing Least General Generalizations Through Anti-Unification}
\label{sec:antiunif}

Consider the ellipsis expression \prog{[f x1, \E, f xn]} (the RHS of \prog{map}).
The expression denotes a list of values obtained by applying a function \prog{f} to each of the elements of the list \prog{x}. What all the elements in the resulting list have in common is the application of \prog{f}, and they (generally) differ in the argument for \prog{f}.
The list can be computed by mapping the function \prog{\BS{}v.f v} over \prog{x}.

Now, consider the RHS of the following definition.

\begin{programCmd}
diffs [x1, \E, xn] = [x1-x2, \E, x\BL{}n-1\BR-xn]
\end{programCmd}

Here, the resulting values are obtained by subtracting from each element of \prog{x} the following element in the list. What all the elements in the result have in common is the application of \prog{(-)}, but they (generally) differ in both arguments.
This list can be computed by zipping \prog{init x} and \prog{tail x} with the function \prog{\BS{}v1 v2.v1-v2}.

Since \prog{map} is \prog{zipWith}$_k$ for $k=1$, the two examples indicate that we generally zip $k$ given (sub)lists with a function that keeps the common parts of the span's two boundary expressions and uses parameters for those parts that vary with the elements of the list(s). 
In particular, the number $k$ of lists equals the number of parameters of the function.

In the \prog{map} example the boundary expressions are $e_1=\prog{f x1}$ and $e_2=\prog{f xn}$, the common part of \e1 and \e2 is \prog{f}, and the difference between \e1 and \e2 is \prog{x1} and \prog{xn}, which is captured by the function \prog{\BS{}v.f v}.
In the \prog{diffs} example the boundary expressions are $e_1=\prog{x1-x2}$ and $e_2=\prog{x\BL{}n-1\BR-xn}$, the common part of \e1 and \e2 is \prog{-}, and the differences between \e1 and \e2 are \prog{x1}, \prog{x2} and \prog{xn}, \prog{x\BL{}n-1\BR}, captured by the function \prog{\BS{}v1 v2.v1-v2}.

In general, we therefore need to compute a function that retains the common parts of two expressions $e_1$ and $e_2$ and represents the differences between them as parameters.
This task can be accomplished by so-called \emph{anti-unification} \cite{AntiUnificationP,AntiUnificationR}, which formally computes the \emph{least general generalization} (LGG) $e$ of two expressions \e1 and \e2. 
The least general generalization captures the minimal differences between two expressions and provides a function that interprets the difference as an interpolation and transition between them.
The differences between \e1 and \e2 are represented in $e$ by variables along with a pair of substitutions $\subst_1$ and $\subst_2$ such that $\subst_1(\expr)=\e1$ and $\subst_2(\expr)=\e2$.
Note that we only need to define anti-unification over boundary expressions, since these are the ones marking the boundaries of spans and fold expressions.
Since the two substitutions have the same domain, they can be succinctly represented by a single mapping $\Var\to\BExp\times\BExp$, which we also call a \emph{pair substitution}.  
\[
\theta \in \Subst \ ::=\ \set{\sequ{x\mapsto(\bexp,\bexp)}{,\,}}
\]
The two individual substitutions can be obtained in the obvious way, that is, when $\subst(x)=(\bexp_1,\bexp_2)$, then $\subst_1(x) = \bexp_1$ and  $\subst_2(x) = \bexp_2$.
The application of a pair substitution to an expression creates a pair of expressions, that is, $\subst(e)=(\subst_1(e),\subst_2(e))$.
The result of anti-unification must be \emph{minimally general} in the sense that it can be obtained from any other generalization by a substitution, that is, for every $e'$ that is a generalization of $\bexp_1$ and $\bexp_2$ there is a plain (that is, non-pair) substitution $\subst'$ with $\subst'(e')=e$.


We can formalize anti-unification through the following judgment that maps a substitution and two expressions to the generalized expression plus a new substitution.
\[
\gen{\bexp}{\bexp}{\bexp}\subst\subseteq 
Subst\times\BExp\times\BExp\times\BExp\times\Subst
\]
%
%
The rules for anti-unification are shown in Figure \ref{fig:antiunif}.
We stipulate a mapping $\gamma$ that maps each expression construct to a fixed, separate number such that $\gamma(e)=\gamma(e')$ iff $e$ and $e'$ are the same syntactic construct (with possibly different parts); $\gamma$ is used in rule \textsc{\textrm{New}} to identify cases when two expressions are different at the their root, which requires a new entry in the computed pair substitution.
Since each expression is its own LGG, two expressions that are equal don't require any generalization (rule \textsc{Equal}). Rule \textsc{Old} covers the case when two expressions have already been generalized, which is reflected by a corresponding entry in the pair substitution.
We use $\gamma(e)$ in rule \textsc{New} to determine an expression's syntactic form and to decide when two expressions differ in their outermost constructor, in which case they need to be replaced by a variable in their LGG.
Note that rule \textsc{New} covers rules for constants and variables.
%
%
The remaining rules work by recursively anti-unifying subexpressions for expressions of the same syntactic form.
Note that we exclude a rule for the anti-unification of (expression arguments of) ellipsis variables to ensure that different ellipsis variables are identified as different variables and get recorded in the computed substitution.

Anti-unification has the following two properties. Note that the equality in the first lemma holds up only to $\alpha$-equality (which is due to the renaming in the \textsc{Abs} rule).

\begin{figure}[t]
\centering
\begin{mathpar}
\infer[Equal]{ }
      {\gen{\bexp}{\bexp}{\bexp}{\subst}}

\infer[Old]{x\mapsto (\bexp_1,\bexp_2)\in\subst}
      {\gen{\bexp_1}{\bexp_2}{x}{\subst}}

\infer[New]{\gamma(\bexp_1)\neq\gamma(\bexp_2)\\
                  (\bexp_1,\bexp_2)\not\in\rng{\subst} \\ x\mathrm{\ fresh}}
      {\gen{\bexp_1}{\bexp_2}{x}{\subst\cup\set{x\mapsto (\bexp_1,\bexp_2)}}}

\infer[Abs]{\gen{\beL}{\sub\xL\xR\beR}{\bexp}{\subst'}  }
      {\gen{\abs{\xL}{\beL}}{\abs{\xR}{\beR}}{\abs{\xL}{\bexp}}{\subst'}}

\infer[App]{\gen{\lft{\bar{f}}}{\rgt{\bar{f}}}{\bar{f}}{\subst'} \\
       \gen[\subst']{\beL}{\beR}{\bexp}{\subst''}  }
      {\gen{\lft{\bar{f}}\;\beL}{\rgt{\bar{f}}\;\beR}{\bar{f}\;\bexp}{\subst''}}
%
%
%
%
\end{mathpar}
\caption{Anti-unification. Since we use anti-unification only for boundary expressions, we only need anti-unification for abstraction and application. A rule for ellipsis variables is also omitted, since it would interfere with deciding whether a list traversal or the generation of a range is required.}
\label{fig:antiunif}
\end{figure}


\begin{lemma}[Correctness]
\label{lem:correct}
If \gen[\varnothing]{\bexp_1}{\bexp_2}{\bexp}{\subst}, then $\subst(\bexp)=(\bexp_1,\bexp_2)$.
\end{lemma}

\noindent
The proofs to this and the following lemma are given in the Appendix.

Lemma \ref{lem:correct} assures us that the anti-unification rules produce correct generalizations for two expressions $\bexp_1$ and $\bexp_2$. In addition, we would like these generalizations also to be minimal so that they represent a tight characterization of the intermediate sequence values between $\bexp_1$ and $\bexp_2$.
Specifically, since the translation creates, in general, a list-producing expression for each substitution computed by anti-unification, we should avoid unnecessary such substitutions.
For example, for the boundary expressions from the RHS of \prog{map} we obtain \gen[\varnothing]{\prog{f x1}}{\prog{f xn}}{\prog{f v}}{\set{\prog{v}\mapsto (\prog{x1},\prog{xn})}} from the anti-unification rules, which prompts the generation of one list. 
However, while the expression \prog{u v} together with the substitution \set{\prog{u}\mapsto (\prog{f},\prog{f}), \prog{v}\mapsto (\prog{x1},\prog{xn})} is also a valid generalization, it is more general than what we need.

The following lemma tells us that the anti-unification rules actually compute \emph{least general} generalizations.

\newcommand\gexp{\ensuremath{\tilde{e}}}
\newcommand\gsubst{\ensuremath{\tilde{\subst}}}
\newcommand\dom{\mathit{dom}}

\begin{lemma}[Minimality]
\label{lem:minimal}
If \gen[\varnothing]{\bexp_1}{\bexp_2}{\bexp}{\subst}, then for all $\gexp,\gsubst$ with $\gsubst(\gexp)=(\bexp_1,\bexp_2)$ there is a plain (that is, non-pair) substitution $\subst^*$ such that $\subst^*(\gexp)=\bexp$.
\end{lemma}

\noindent
With anti-unification we can now show how to translate expressions that employ ellipsis notation into expressions without ellipses.

\subsection{Translation Semantics}
\label{sec:translate}
\label{sec:semantics}

The transformation of lists specified by ellipsis expressions can be explained through a translation into the core language. These translation rules, shown in Figure \ref{fig:semantics}, are driven by the following pieces of information.

\begin{enuma}
\item \textit{The origin and parts of the lists to be transformed.} Which ellipsis expressions denote sublists, and which generate a range of values of an enumeration type?
~Do these lists require reversal?
These decisions are captured by the rules \textsc{Sublist} and \textsc{Range}, supported by the corresponding auxiliary functions \prog{fwd}, \prog{bwd}, \prog{dyn}, and \prog{range}.

\item
\textit{The number of input lists to be transformed.}
This information is gleaned from the substitutions that are generated by the anti-unification premise of rule \textsc{Span}. Each substitution calls for a separate input list, and the number $k$ of substitutions generates a call to \prog{zipWith}$_k$ with a $k$-ary function as a parameter.

\item \textit{The elements to be generated to fill the ellipsis.} This information is captured by the least general generalization produced in the premise of \textsc{Span}.

\item \textit{The potential aggregation of the generated list.} The generated list elements may be subject to aggregation with a binary operation, described by rules \textsc{Fold} and \textsc{Fold1}, or selection, described by rule \textsc{EVar}.
\end{enuma}

\newcommand\tagCR{\ensuremath{\ast}}

%
Before presenting the translation rules, we illustrate the translation with several examples.
Since the translation operates on the syntax of the extended core language defined in Figures \ref{fig:syncore} and \ref{fig:syntax}, we replace equations with lambda abstractions and case expressions.

Consider the simple case in which a function employs a single case rule with an ellipsis pattern 
and an ellipsis expression \eexp.
\begin{align}
\prog{\BS{}x.case x of \BL{}[x1, }\E\prog{, xn] -> }\eexp\BR \tag{\tagCR}
\end{align}
We can have three different cases for the ellipsis expression \eexp. 

First, consider the definition of \prog{map} in which \eexp\ is \prog{[f x1, \E, f xn]}, which means we have no aggregation or selection (information piece (d)).

\begin{programCmd}
map = \BS{}f.\BS{}x.case x of \BL{}[x1, \E, xn] -> [f x1, \E, f xn]\BR
\end{programCmd}

As we have already discussed, anti-unification produces the function \prog{\BS{}v.f v} as the LGG (information piece (c)) together with the substitution \set{\prog{v}\mapsto (\prog{x1},\prog{xn})}. From that we can infer that we are dealing with one input list, that is, $k=1$ (information piece (b)). We also know that the list is given through a binding of the case rule's LHS, and the indices \prog{1} and \prog{n} of the ellipsis variables tell us that the whole of the input list is to be processed in the original order (information piece (a)). 
Altogether, we obtain the following translation.

\begin{program}
map = \f.\x.case x of {x -> let n = length x in
                            zipWith1 (\v.f v) (fwd 1 n x)}
\end{program}


The function \prog{fwd} extracts the required sublist. In this specific instance, it simply rebuilds \prog{x}. We will discuss the definition of \prog{fwd} and its inference later.
The binding for \prog{n} is created to avoid recomputing the length of the list in cases where \prog{n} is referenced in different subexpressions, as illustrated by the next example.

Applying some basic algebraic optimizations to the translations can simplify the resulting definition. For example, by applying the identity $\prog{fwd 1 n x}=\prog{x}$ the unreferenced binding for \prog{n} can be eliminated (although that would not be critical in a lazy evaluation setting). Then we can simplify further by applying the identity $\prog{case x of x ->~}e=e$. 
This simplification is applicable to all function definitions consisting of a single equation, and we will apply it from now on implicitly.  
Finally, three eta-reductions yield the following.

\begin{program}
map = zipWith1
\end{program}

As another example, consider the function \prog{enumerate}, defined in Figure \ref{fig:elliex}. Again, this does not involve aggregation or selection.
Here, anti-unification of $e_1=\prog{(1,x1)}$ and $e_2=\prog{(n,xn)}$ yields the LGG $f=\prog{\BS v1 v2.(v1,v2)}$, a function of two parameters, which means that the translation will generate an application of \prog{zipWith2}, which applies \prog{f} to two lists. The first of these lists computes a range of numbers (using the auxiliary function \prog{range}), and the second list is the same as in the case of \prog{map}, which leads to the following translation (again, simplified by eliminating the unnecessary one-rule case expression).

%

\begin{program}
enumerate = \x.let n = length x in 
               zipWith2 (\v1 v2.(v1,v2)) (range 1 n) (fwd 1 n x)
\end{program}

This definition can be simplified only slightly. While we can eliminate the application of \prog{fwd}, we have to keep the binding for \prog{n}.

\begin{program}
enumerate = \x.let n = length x in zipWith2 (\v1 v2.(v1,v2)) (range 1 n) x
\end{program}


The basic form illustrated by \prog{map} and \prog{enumerate} can be generalized in several different ways. 
First, the indices in the ellipsis variables are not restricted to \prog{1} and \prog{n} (see, for example, \prog{tail} and \prog{init}), and they can occur in a different order, as in \prog{reverse}, which then generates a call to \prog{bwd} instead of \prog{fwd}.
Second, ellipsis expressions may contain different ellipsis variables that refer to different list arguments (as in \prog{zip} and \prog{append}).
And finally, an ellipsis expression can have multiple segments, which are translated individually and then concatenated. 
Examples for this are the definition \prog{append} and the following definition of the function \prog{rotateLeft}.

\begin{programCmd}
rotateLeft k [x1, \E, xn] = [x\BL{}k+1\BR, \E, xn, x1, \E, xk]    
\end{programCmd}

The ellipsis expression on the RHS contains two segments of the same list, each of which is translated into an application of \prog{zipWith1 (\BS{}v.v)}, since the list elements are not transformed in any way.
And since the indices indicate a forward traversal for in both segments, the arguments to \prog{zipWith1} are two \prog{fwd} expressions with the corresponding indices.
Thus the translation of \prog{rotateLeft} will produce the following.

\begin{program}
rotateLeft = \k.\x.let n = length x in 
                   zipWith1 (\v.v) (fwd (k+1) n x) ++ 
                   zipWith1 (\v.v) (fwd 1 k x)
           = \k.\x.let n = length x in fwd (k+1) n x ++ fwd 1 k x
\end{program}

This example also illustrates how sublists are extracted via indices of ellipsis variables that become arguments of traversal functions (\prog{fwd} in this case).

The translation of ellipsis expressions with aggregations proceeds in the same way as for the previous examples. The only difference occurs at the very end where we generate an additional fold operation. 
For example, the \prog{fac} and \prog{sum} definitions fit the pattern (\tagCR) where \eexp\ is 
\prog{1 * \E{} * n} and \prog{x1 + \E{} + xn}, respectively.
The translations thus produce applications of \prog{zipWith1} together with \prog{range} (in the case of \prog{fac}) and \prog{fwd} (in the case of \prog{sum}) plus the additional folds.

\begin{program}
fac = \n.foldl1 (*) (zipWith1 (\v.v) (range 1 n))
sum = \x.let n = length x in foldl1 (+) (zipWith1 (\v.v) (fwd 1 n x))
\end{program}

As with the previous examples, these definitions can be simplified.

\begin{program}
fac = \n.foldl1 (*) (range 1 n)
sum = foldl1 (+)
\end{program}


Finally, a selection is given by an ellipsis expression \eexp\ of the form \prog{x}\BL\bexp\BR, which is translated directly into an application of the list index (or subscript) operation \prog{!!}, offset by one, that is, {\prog{x}\BL\bexp\BR} translates to \prog{x!!(}\bexp\prog{-1)}. For example, the translation of \prog{head} leads to the following definition.

\begin{program}
head = \x.x!!0
\end{program}

Of course, since the RHS of a case rule can be an arbitrary expression $e$, we can combine ellipsis expressions in a variety of different ways. Consider the following example.





\begin{programCmd}
variance []           _ = 0
variance [x1, \E, xn] a = (sqr (x1-a) + \E + sqr (xn-a)) / n
\end{programCmd}

The congruence rules of the translation (see the top part of Figure \ref{fig:semantics}) facilitate the translation of ellipsis expressions in the context of other expressions, in this case as the dividend of ``\prog{/ n}''.

\begin{programCmd}
variance = \BS{}x.\BS{}a.case x of
    [] -> 0
    x  -> let n = length x in
              foldl1 (+) (zipWith1 (\BS{}v1.sqr (v1-a)) (fwd 1 n x)) / n
\end{programCmd}








\def\skipv{0pt}
\def\raisev{0pt}
\def\topv{3ex}
\def\botv{0pt}

\begin{figure}[t]
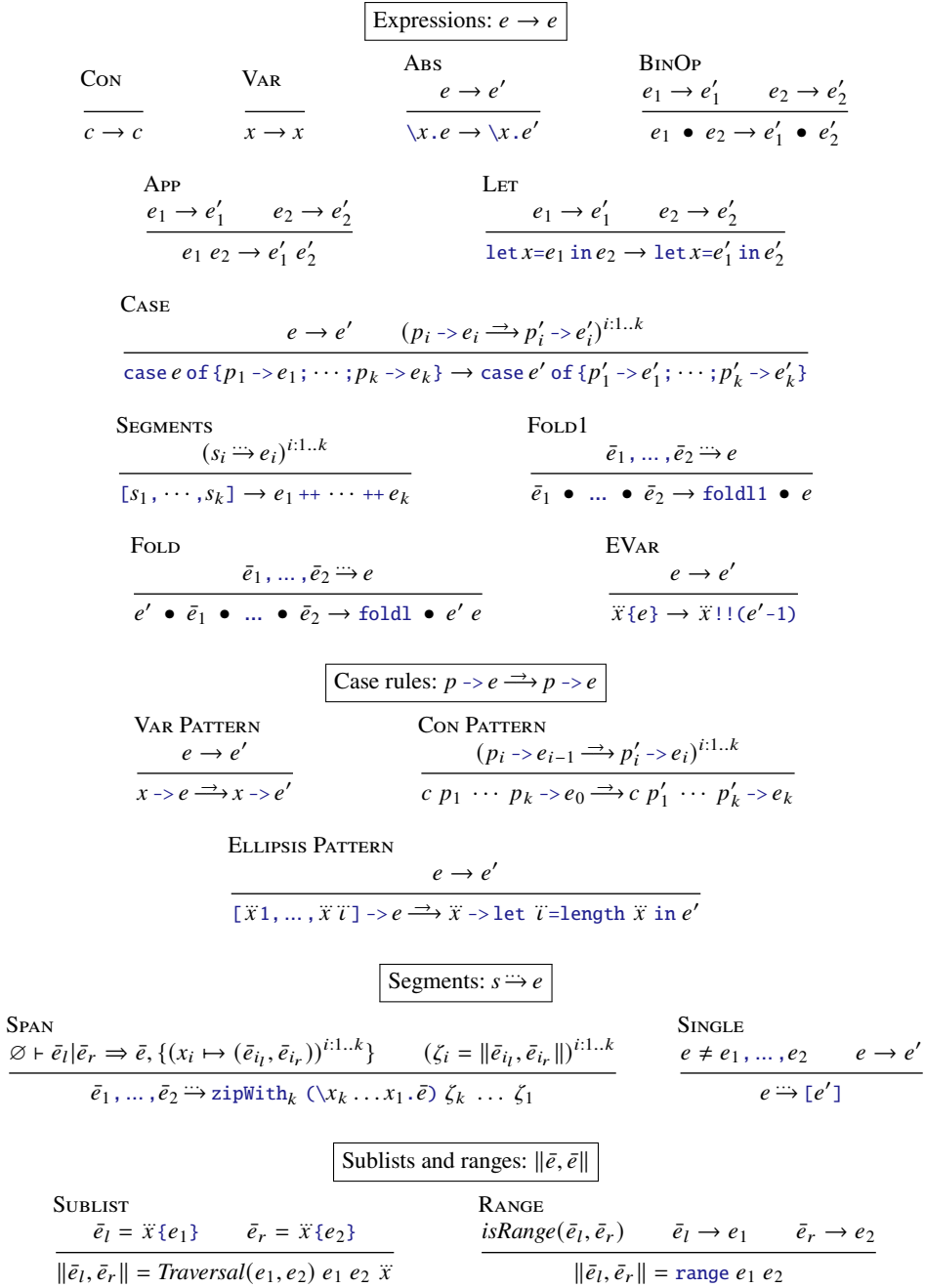

\centering\small
\begin{mathpar}
\raisebox{\raisev}[\topv][\botv]{\fbox{Expressions: \trE{e}{e}}}
\\
\infer[Con]{ }{\trE{c}{c}}

\infer[Var]{ }{\trE{x}{x}}

\infer[Abs]
    {\trE{e}{e'}}
    {\trE{\abs{x}{e}}{\abs{x}{e'}}}

\infer[BinOp]
    {\trE{\eL}{e'_1} \\ \trE{\eR}{e'_2}}
    {\trE{\eL\ \op\ \eR}{e'_1\ \op\ e'_2}}
\\
\infer[App]
    {\trE{\eL}{e'_1} \\ \trE{\eR}{e'_2}}
    {\trE{\eL\;\eR}{e'_1\;e'_2}}

\infer[Let]
    {\trE{\eL}{e'_1} \\ \trE{\eR}{e'_2}}
    {\trE{\xlet{x}{\eL}{\eR}}{\xlet{x}{e'_1}{e'_2}}}
    
\infer[\lCase]
    {\trE{e}{e'} \\
     (\trCaseRule{\xrule{p_i}{e_i}}{\xrule{p'_i}{e'_i}})^{i:1..k}}
    {\trE{\case{e}{\sequx{\xrule{p_1}{e_1}}{\prog{;}}{\xrule{p_k}{e_k}}}}
         {\case{e'}{\sequx{\xrule{p'_1}{e'_1}}{\prog{;}}{\xrule{p'_k}{e'_k}}}}}

\infer[Segments]{(\trSeg{\seg_i}{e_i})^{i:1..k}}
      {\trE{\prog{[}\sequx{\seg_1}{\prog{,}}{\seg_k}\prog{]}}{\sequx{\eL}{\,\prog{++}\,}{e_k}}}

\infer[Fold1]
      {\trSeg{\bexp_1\prog{,\E,}\bexp_2}{e}}
      {\trE{\bexp_1\ \op\ \E\ \op\ \bexp_2}
             {\prog{foldl1}\ \op\ e}}
             
\infer[Fold]
      {\trSeg{\bexp_1\prog{,\E,}\bexp_2}{e}}
      {\trE{e'\ \op\ \bexp_1\ \op\ \E\ \op\ \bexp_2}
             {\prog{foldl}\ \op\ e'\ e}}

\infer[EVar]{\trE{e}{e'}}
      {\trE{\evar\BL{}e\BR}{\evar\prog{!!(}e'\prog{-1)}}}
\end{mathpar}
\vskip\skipv
\begin{mathpar}
\raisebox{\raisev}[\topv][\botv]{\fbox{Case rules: \trCaseRule{\xrule{p}{e}}{\xrule{p}{e}}}}
\\
\infer[Var Pattern]
    {\trE{e}{e'}}
    {\trCaseRule{\xrule{x}{e}}{\xrule{x}{e'}}}

\infer[Con Pattern]
    {(\trCaseRule{\xrule{p_i}{e_{i-1}}}{\xrule{p'_i}{e_i}})^{i:1..k}}
    {\trCaseRule{\xrule{c\ \sequx{p_1}{\ }{p_k}}{e}_0}
         {\xrule{c\ \sequx{p'_1}{\ }{p'_k}}{e_k}}}

\infer[Ellipsis Pattern]
    {\trE{e}{e'}}
    {\trCaseRule{\xrule{\prog{[}\evar\prog{1}\prog{,\E,}\evar\ivar\prog{]}}{e}}{\xrule{\evar}{\prog{let}\ \ivar\prog{=length}\ \evar\ \prog{in}\ e'}}}
\end{mathpar}
\vskip\skipv
\begin{mathpar}
\raisebox{\raisev}[\topv][\botv]{\fbox{Segments: \trSeg{\seg}{e}}}
\\
\infer[Span]
      {\gen[\varnothing]{\bexp_l}{\bexp_r}{\bexp}{\set{(x_i\mapsto (\bexp_{i_l},\bexp_{i_r}))^{i:1..k}}} \\
       (\zeta_i=\sliceOrRange{\bexp_{i_l}}{\bexp_{i_r}})^{i:1..k}}
      {\trSeg{\bexp_1\prog{,\E,}\bexp_2}
             {\prog{zipWith}_k\ \prog{(}\abs{x_k\ldots x_1}{\bexp}\prog{)}\ \zeta_k\ \ldots\ \zeta_1}}

\infer[Single]
      {e\neq \eL\prog{,\E,}\eR
      \\
        \trE{e}{e'}}
      {\trSeg{e}{\prog{[}e'\prog{]}}}
\end{mathpar}
\vskip\skipv
\begin{mathpar}
\raisebox{\raisev}[\topv][\botv]{\fbox{Sublists and ranges: \sliceOrRange{\bexp}{\bexp}}}
\\
\infer[Sublist]{\bexpa=\evar\BL{}e_1\BR \\ \bexpb=\evar\BL{}e_2\BR}
      {\sliceOrRange{\bexpa}{\bexpb} = {\traversal{e_1}{e_2}\ e_1\ e_2\ \evar}}

\infer[Range]{\isRange{\bexpa}{\bexpb} \\ \trE\bexpa{e_1} \\ \trE\bexpb{e_2}}
      {\sliceOrRange{\bexpa}{\bexpb} = {\prog{range}\ e_1\ e_2}}

%
%
\end{mathpar}
\caption{Translation semantics. 
} 
\label{fig:semantics}
\end{figure}

The rules for translating programs with ellipsis patterns and expressions shown in Figure \ref{fig:semantics} are grouped into four sections.

The first part consists of rules of the form \trE{e}{e} that describe how to generate expressions in the core language by eliminating ellipsis patterns and expressions.
The first seven of these rules simply cover the base cases of constants and variables and provide congruence rules for the standard language constructs.
The last of these rules, \textsc{Case}, translates case rules in its second premise, described in the second section of Figure \ref{fig:semantics}, discussed below.
The final four rules in this group define the translation of ellipsis expressions.
First, a list of segments is translated into the concatenation of the translation of each segment (rule \textsc{Segments}). 
Second and third, fold expressions are translated by first translating their span arguments and then simply adding a \prog{foldl} or \prog{foldl1} application to the resulting $\prog{zipWith}_k$ expression (rules \textsc{Fold} and \textsc{Fold1}).
These three rules make use of rules for translating segments that are described in the third section of Figure \ref{fig:semantics}.
Fourth, rule \textsc{EVar} translates ellipsis variables into list indexing expressions.

The second group of rules for translating case rules have the form \trCaseRule{\xrule{p}{e}}{\xrule{p}{e}} and consists of three rules. 
The rule \textsc{Var} keeps single variables as patterns unchanged and simply translates the RHS of the case rule.
Rule \textsc{Con Pattern} translates nested patterns while also accumulating translated RHS expressions.
The final rule in this group, \textsc{Ellipsis Pattern}, shows how a case rule with an ellipsis pattern is translated into a corresponding rule with an additional binding for the length of the matched list added to the RHS. 

To illustrate the point of the rule \textsc{Con Pattern}, consider the function \prog{append}, shown here for instructional purposes as defined by a case rule using a tuple.

\begin{programCmd}
append = \BS{}l1.\BS{}l2.case (l1,l2) of 
                  \BL([x1, \E, xn],[y1, \E, ym]) -> [x1, \E, xn, y1, \E, ym]\BR
\end{programCmd}

In the conclusion of \textsc{Con Pattern}, $c$ is the tuple constructor taking arguments $p_1=\prog{[x1, \E, xn]}$ and $p_2=\prog{[y1, \E, ym]}$, with $e_0=\prog{[x1, \E, xn, y1, \E, ym]}$ on the RHS.
Since $c$ takes two arguments, the rule instantiates $k=2$ premises. The first premise for $i=1$ is obtained by an application of rule \textsc{Ellipsis Pattern} to \xrule{p_1}{e_0}, resulting in the following intermediate case rule where $p_1'=\prog{x}$ and $e_1$ is the translation of $e_0$, that is, the concatenation of the two \prog{zipWith1} expressions.

\begin{programCmd}
(x,[y1,\E,ym]) -> let n=length x in 
                  zipWith1 (\BS{}v.v) (fwd 1 n x) ++
                  zipWith1 (\BS{}v.v) (fwd 1 m y)
\end{programCmd}

In the second premise (for $i=2$), rule \textsc{Ellipsis Pattern} is applied once more, now to \xrule{p_2}{e_1}, which yields $p_2'=\prog{y}$ and $e_2=\prog{let m=length y in }e_1$.
This completes the translation, producing the following case rule.

\begin{programCmd}
(x,y) -> let m=length y in 
         let n=length x in
         zipWith1 (\BS{}v.v) (fwd 1 n x) ++
         zipWith1 (\BS{}v.v) (fwd 1 m y)
\end{programCmd}

This example demonstrates the cumulative behavior of rule \textsc{Con Pattern}, which enables the retention of bindings introduced by multiple ellipsis patterns.

The two rules in the third group for translating segments have the form \trSeg{s}{e}. Rule \textsc{Single} translates expressions that denote individual list values into singleton lists. Rule \textsc{Span} translates spans into applications of $\prog{zipWith}_k$, employing anti-unification as illustrated in detail earlier. 

Finally, the rules in the fourth group, \textsc{Sublist} and \textsc{Range}, translate pairs of expressions obtained by anti-unification into an expression that computes either a sublist or a range.
This decision is based on whether the substitution entry from anti-unification is a pair of ellipsis variables, which demand a sublist, or any other pair of expressions, which call for a range of values from an enumerated type.
In other words, rule \textsc{Range} is selected whenever at least one of the expressions is \emph{not} an ellipsis variable, a condition encapsulated in the predicate \rangeName, which is defined as follows.
\[
\isRange\bexpa\bexpb :\Leftrightarrow \forall \evar,\bexp. \bexpa\neq\evar\BL{}\bexp\BR \vee \bexpb\neq\evar\BL{}\bexp\BR
\]
In this case the translation generates an application of the function \prog{range}, which is a wrapper around the already existing ellipsis notation in Haskell.

\begin{program}
range l r = [l, if r>l then l+1 else l-1..r]
\end{program}

Rule \textsc{Sublist} generates different list traversals based on the relationship of the indexing expressions $e_1$ and $e_2$. In general, the traversal has to decide whether to retain or reverse the order of elements in the produced list and also whether that list should be empty.
This is done as follows.
For \e1 and \e2 we determine the ``relative location'' they indicate in the list, which is one of $\{\Lft,\Mdl,\Rgt\}$ where $\Lft<\Mdl<\Rgt$.
This determination is based on the set $V(e_1)$ and $V(e_2)$ of variables they reference, which may include index variables ($\ivar\in\IVar$) that indicate the lengths of lists, such as \prog{n} or \prog{m}, but also variables that are bound outside of the ellipsis pattern, such as \prog{k} in the definition of \prog{drop} or \prog{rotateLeft}. 
Index variables indicate \Rgt\ as a list location, other variables indicate \Mdl, and an indexing expression with no variables indicate \Lft\ as a list location.
\[
\Loc(e) = \left\{\begin{array}{@{}l@{\quad}l}
\Lft & \mathrm{if\ }V(e)=\varnothing \\
\Mdl & \mathrm{if\ }\not\exists\ivar\in V(e) \wedge V(e)\neq\varnothing \\
\Rgt & \mathrm{if\ }\exists\ivar\in V(e)
\end{array}\right.
\]
With the help of \Loc\ we can determine the relative position of the indexing expressions \e1 and \e2 with respect to one another. 
With that we can decide (a) which direction the list has to be traversed and (b) whether the sublist is empty. Specifically, if $\Loc(e_1)<\Loc(e_2)$, we know statically that the list should be traversed in the forward direction, whereas if $\Loc(e_1)>\Loc(e_2)$, the elements are to be visited in reverse order. If $\Loc(e_1)=\Loc(e_2)$, the traversal direction is determined dynamically based on the values of \e1 and \e2. The function \traversalName\ maps index expressions based on their relative location to functions that implement the required traversal.
\[
\traversal{e_1}{e_2} = \left\{\begin{array}{@{}l@{\quad}l}
\prog{fwd} & \mathrm{if\ }\Loc(e_1)<\Loc(e_2) \\
\prog{bwd} & \mathrm{if\ }\Loc(e_1)>\Loc(e_2) \\
\prog{dyn} & \mathrm{if\ }\Loc(e_1)=\Loc(e_2)
\end{array}\right.
\]
The following table illustrates the different cases that might occur.
%
%
It shows the expected results of an ellipsis expression \prog{[x\BL}$e_1$\prog{\BR,\E,x\BL}$e_2$\prog{\BR]} for the list $\prog{x} = \prog{[1,2,3,4,5]}$, depending on the relationship between $\Loc(e_1)$ and $\Loc(e_2)$ and concrete values \sem{e_1} and \sem{e_2} of the index expressions \e1 and \e2.
Let \prog{n} be an index variable, and let \prog{k} and \prog{l} be plain variables.
\[
\begin{array}{cccc|cc|l}
\multicolumn{3}{c}{\emph{Index expressions}} & &&
\\
\multicolumn{3}{c}{\emph{and their relationship}} & \traversalName & 
\multicolumn{2}{c|}{\emph{Index values}} & \emph{Expected results}
\\
e_1 & e_2 & \cmp{e_1}{e_2} & & \sem{e_1} & \sem{e_2} 
& \prog{[x\BL}e_1\prog{\BR,\E,x\BL}e_2\prog{\BR]} \\
\hline
\prog{k} & \prog{n} & < & \prog{fwd} 
   & 2 & 4 & \prog{[2,3,4]}\\
&&&& 3 & 3 & \prog{[3]}\\
&&&& 4 & 2 & \prog{[]}\\
\hline
\prog{n} & \prog{k} & > & \prog{bwd} 
   & 2 & 4 & \prog{[]}\\
&&&& 3 & 3 & \prog{[3]}\\
&&&& 4 & 2 & \prog{[4,3,2]}\\
\hline
\prog{k} & \prog{l} & = & \prog{dyn} 
   & 2 & 4 & \prog{[2,3,4]}\\
&&&& 3 & 3 & \prog{[3]}\\
&&&& 4 & 2 & \prog{[4,3,2]}\\
\end{array}
\]
We conclude with the definitions of the functions \prog{fwd}, \prog{bwd}, and \prog{dyn}, which implement the desired behavior.

\begin{program}
fwd l r = take (r-l+1) . drop (l-1)
bwd l r = reverse . fwd r l
dyn l r | l<=r      = fwd l r
        | otherwise = bwd l r
\end{program}





\section{Split Patterns}
\label{sec:splitpatterns}

When programming with ellipsis notation, most of the work is done by ellipsis expressions. 
It is striking that an ellipsis pattern essentially consists of just two variables, \evar\ and \ivar---the rest is just syntactic embellishment for triggering the ellipsis intuition of the programmer.
%
%
However, we can enhance the expressiveness of ellipsis notation through a minor extension of ellipsis patterns that has a computational interpretation.
To this end, we introduce the concept of \emph{split patterns}, which essentially consist of two concatenated ellipsis patterns plus an attached expression. 
Split patterns have the following syntax.
\[
\begin{array}[b]{r@{\ }c@{\ }l@{\ \ }c@{\ \ }l}
\epat \IN \EPat \IS \ldots \OR 
   \prog{[}\evar\pI\prog{,\E,}\evar\ivar\prog{,}
           \evar[y]\pI\prog{|}e\prog{,\E,}\evar[y]\ivar[\kappa]\prog{]}
\end{array}
\]
The idea of split patterns is to split lists into two segments based on a predicate. Specifically, when a split pattern is matched against a list, the variable \ivar\ is bound to the largest value such that the expression $e$ evaluates to \prog{False} for all elements \evar\prog{1} through \evar\ivar. This also means that the expression evaluates to \prog{True} for the element \evar[y]\prog{1} (if any such element exists in the list).

Here is a concrete example of a split pattern for matching the list of even numbers at the beginning of a list.

\begin{programCmd}
[x1, \E, xk, y1|odd, \E, yn]
\end{programCmd}

The list variables (here \prog{x} and \prog{y}) as well as the index variables (here \prog{k} and \prog{n}) of the two ellipsis patterns must be different, and the expression (here \prog{odd}) associated with \prog{y1} should be of type \prog{a -> Bool} where \prog{a} is the type of the list elements. 
The semantics of split patterns suggests the following syntactic sugar, which can be convenient in cases when the negated predicate is easier to express or if its use is more intuitive.

\begin{programCmd}
[x1, \E, xk, y1|p, \E, yn]  =  [x1, \E, xk|not.p, y1 \E, yn]
\end{programCmd}

Thus if the previous example seems like an odd way of expressing matching even numbers, we could instead use the following alternative pattern.

\begin{programCmd}
[x1, \E, xk|even, y1 \E, yn]
\end{programCmd}

A split pattern basically realizes the idea of Haskell's \prog{break} function, which splits a list into two lists based on a predicate, so that the first list contains the prefix of the list for which the predicate is not true and the second list starts with an element for which the predicate is true.  
In fact, \prog{break} can be defined with the help of a split pattern as follows .

\begin{programCmd}
break p [x1, \E, xk, y1|p, \E, yn] = ([x1, \E, xk],[y1, \E, yn])
\end{programCmd}

Of course, a split pattern has other uses as well. For example, we can define the function \prog{filter} as follows.

\begin{programCmd}
filter _ []                         = []
filter p [x1, \E, xk|p, y1, \E, yn] = [x1, \E, xk] ++ filter p [y2, \E, yn]
\end{programCmd}

Note that \prog{y1} has been dropped in the recursive call. Since \prog{y1} is the first element for which \prog{p} is false, a perhaps clearer version of the second equation uses \prog{not.p} in the second list.

\begin{programCmd}
filter p [x1, \E, xk, y1|not.p, \E, yn] = [x1, \E, xk]++filter p [y2, \E, yn]
\end{programCmd}

Another use of a split pattern is locating elements with particular properties. The function \prog{findIndex} provides a most basic example.

\begin{programCmd}
findIndex p [x1, \E, xk, y1|p, \E, yn] = if n==0 then Nothing else Just k
\end{programCmd}

As a concluding example, we define the merge sort algorithm in terms of a helper function \prog{merge}, which employs a split pattern.

\begin{programCmd}
mergesort []           = []
mergesort [x]          = [x]
mergesort [x1, \E, xn] = mergesort [x1,           \E, x\{div n 2\}] `merge` 
                         mergesort [x\{div n 2+1\}, \E, xn]

merge [] ys = ys
merge xs [] = xs
merge [x1, \E, xn] [y1, \E, yk|(<=x1),z1, \E, zm] =
                   [y1, \E, yk] ++ merge [z1, \E, zm] [x1, \E, xn]
\end{programCmd}

The rule in Figure \ref{fig:splitpattern} shows how split patterns are translated into an application of the predefined function \prog{break}. (We use as syntactic sugar a pair of variables in a lambda abstraction to avoid the notational overhead of an additional case expression.)

\begin{figure}[t]
\centering
\begin{mathpar}
\infer[Split Pattern]
{\trCaseRule{
  \xrule{\prog{([}
    \evar\prog{1,\E,}\evar\ivar
  \prog{],[} 
    \evar[y]\prog{1,\E,}\evar[y]\ivar[\kappa]
  \prog{])}}{e}
  }
  {\xrule{\prog{(}\evar\prog{,}\evar[y]\prog{)}}{e'}} \\
  \trE{f}{f'} \\
  \prog{z}\ \mathrm{fresh}}
{\trCaseRule{
 \xrule{\prog{[}
    \evar\prog{1,\E,}\evar\ivar
  \prog{,} 
    \evar[y]\prog{1|}f\prog{,\E,}\evar[y]\ivar[\kappa]
  \prog{]}}{e}
  }
  {\xrule{\prog{z}}{\prog{(\BS(}\evar\prog{,}\evar[y]\prog{).}e'\prog{) (break}\ f'\ \prog{ z)}}}
}

%
\end{mathpar}
\caption{Translation semantics of split patterns. A split pattern behaves the same as two ellipsis patterns applied to a list split in two by the function \prog{break}.
} 
\label{fig:splitpattern}
\end{figure}

\section{Evaluation}
\label{sec:eval}

The \prog{Data.List} module of Haskell's standard base library provides a set of fundamental operations on lists, such as \prog{map} and \prog{filter}.
These functions cover a wide range of functionality on lists and can thus be used as a basis for answering the question of how applicable ellipsis notation is to operations on lists in general.

Figure \ref{fig:data_list_table} depicts all categories described by section headers within the documentation of \prog{Data.List} \cite{DataList}. Alongside each category, the total number of evaluated functions is indicated in addition to how many can be reasonably definable using ellipses.
In this evaluation, a ``reasonable'' definition constitutes a usage of ellipses that preserve or improve the clarity of the definition, in which there exists no clearer definition with a preexisting notation such as list comprehensions.
Furthermore, we have omitted some functions from this analysis (for example versions of $\prog{zip}_k$ and $\prog{unzip}_k$ for $k>2$) that would not contribute any additional tests for the expressiveness of ellipsis notation.


Bolstering this evaluation is the usage of ellipses by \prog{Data.List}'s documentation to describe the behavior of some functions.
For instance, $\prog{map}$ and $\prog{foldr}$ are described as follows.
\begin{programCmd}
map f [x1, x2, \E, xn] == [f x1, f x2, \E, f xn]

foldr f z [x1, x2, \E, xn] == x1 `f` (x2 `f` \E (xn `f` z)\E)
\end{programCmd}

The documentation also takes more liberties with the notation by using \textit{unbounded} ellipses, as in the case of \prog{iterate} and an alternative description of \prog{map}.
\begin{programCmd}
iterate f x == [x, f x, f (f x), \E]

map f [x1, x2, \E] == [f x1, f x2, \E]
\end{programCmd}

\newlength{\digwidth}
\settowidth{\digwidth}{$9$}
\newcommand\bl{\makebox[\digwidth][c]{ }}
\newcommand\blz{\bl\ }

\begin{figure}[t]
\centering
\begin{tabular}{|l|c|c|}
\hline
\textbf{Sections in $\prog{Data.List}$} & \textbf{Standard} & \textbf{Ellipsis Notation} \\
\hline
Basic functions & \bl7 & \bl7 \\
List transformations & \bl7 & \bl4 \\
Reducing lists (folds) & 14 & 12 \\
Building lists & 11 & \blz \\
Sublists & 14 & \bl6 \\
Searching lists & \bl6 & \bl2 \\
Indexing lists & \bl5 & \bl1 \\
Zipping and unzipping lists & \bl3 & \bl3 \\
Special lists & 11 & \bl2\\
Generalized functions & 10 & \bl2\\
\hline
\textbf{Total} & 88 & 39 \\
\hline
\end{tabular}
\caption{Evaluation of the applicability of ellipsis notation to a subset of the \prog{Data.List} module of Haskell's base library. The \textbf{Ellipsis Notation} column shows how many functions of the total are reasonably definable using ellipses. Some operations were omitted from this evaluation---functions that do not take a list as input (such as \prog{singleton}), strict variants of otherwise evaluable functions (such as \prog{foldl'}), and ``\prog{generic}'' operations meant to overload existing functionality (such as \prog{genericLength}) were not considered. Moreover, functions encompassing the same functionality (namely, $\prog{zip}_k$, $\prog{unzip}_k$, and $\prog{zipWith}_k$) were evaluated as one.}
\label{fig:data_list_table}
\end{figure}


While the existing categorization of the $\prog{Data.List}$ functions is quite clear, it is still somewhat difficult to pinpoint where ellipsis notation is applicable in broader categories such as ``Basic functions'' and ``Sublists''.
To this end, Figure \ref{fig:list_comprehensions_table} presents an alternate categorization---greatly inspired by the existing categorization---intended to more precisely convey functional intent.
This reordering allows a more critical analysis of what ellipsis notation excels in representing, as well as its weaknesses.
This analysis also includes ellipsis notation extended by split patterns. 

\newlength{\hdwidth}
\settowidth{\hdwidth}{\textbf{EN + SP}}

\begin{figure}[t]
\centering\small
\begin{tabular}{|l|l|c|c|c|c|}
\hline
&&& 
\textbf{Ellipsis} & 
\textbf{List} & \textbf{Ell. + Split} 
\\ 
\textbf{List operation} & \textbf{Examples} & \textbf{Standard} & 
\textbf{Notation} & 
\textbf{Compr.} & \textbf{Patterns} \\ 
\hline
Mapping & \prog{map}, \prog{zip}$_k$, \prog{zipWith}$_k$ & \bl3 & \win{\bl3} & 1 &   \\  
Transformations & \prog{transpose}, \prog{unzip}$_k$ & \bl6 & \win{\bl3} & 2 &  \\
Reductions & \prog{fold}, \prog{mapAccum} & 20 & \win{16} & 4 &  \\
Constructions & \prog{(++)}, \prog{scan}, \prog{repeat} & 18 & \bl1 & \win{5} & \win{5} \\
Measuring & \prog{null}, \prog{length} & \bl2 & \win{\bl2} & 1 &  \\
Indexing & \prog{elemIndex}, \prog{findIndex} & \bl4 & \blz & \win{4} & \win{4} \\
Sorting & \prog{sort}, \prog{insert} & \bl4 & \blz & \win{4} & \win{4} \\
\hline
\multicolumn{6}{|l|}{Sublist extraction} \\
\hline
...by length & \prog{reverse}, \prog{take}, \prog{drop} & \bl6 & \win{\bl6} & 5 &  \\
...by predicate & \prog{break}, \prog{span} & \bl5 & \bl1 &   & \win{5} \\
\hline
\multicolumn{6}{|l|}{Searching} \\
\hline
...by index & \prog{(!!)}, \prog{head}, \prog{last} & \bl3 & \win{\bl3} & \win{3} &  \\
...by equality & \prog{elem}, \prog{notElem}, \prog{lookup} & \bl3 & \bl2 & \win{3} & \win{3} \\
...by sublist & \prog{isPrefixOf}, \prog{isSuffixOf} & \bl3 & \win{\bl2} &  &  \\
...by predicate & \prog{find}, \prog{filter}, \prog{intersect} & 11 & \blz & 5 & \win{9} \\
\hline
\multicolumn{2}{|l|}{\textbf{Total}} & 88 & 39 & 37 & 65 \\
\hline
\end{tabular}
\caption{Evaluation of the applicability of ellipsis notation alongside list comprehensions (with no extensions such as parallel list comprehensions) to a recategorization of the \prog{Data.List} module of Haskell's base library. Also added are the numbers for ellipsis notation with split patterns. An empty entry in the last column means that the number of reasonably definable functions did not increase, that is, the number is identical to the entry in \textbf{Ellipsis Notation}. This way, it becomes more apparent where additional benefits of split patterns can be expected.}
\label{fig:list_comprehensions_table}
\end{figure}

While it is striking that ellipsis notation is able to express a significant number of $\prog{Data.List}$ functions, the table in Figure \ref{fig:list_comprehensions_table} reveals that, for instance, ellipses are rather unsuccessful in searching lists without extension. 
Similarly, the recategorization also highlights what kind of functions can profit from the extension of split patterns.

As is apparent from the examples in Section \ref{sec:translate}, the rather simple translation that defines the semantics can produce inefficient results for functions using ellipses, in particular, through the generation of calls to \prog{fwd}/\prog{bwd}/\prog{dyn}, \prog{range}, or \prog{length} that could be either combined with one another or avoided altogether.
To get a sense of the overall (in)efficiency we compare the results of the translations with the definitions used in \prog{Data.List} and measure the overhead they incur in terms of additional list traversals.
Figure \ref{fig:overhead} shows in the $n$th row how many functions require $n$ additional list iterations. 
The columns \textbf{Naive Translation} give the numbers for the plain translation, and the column \textbf{Basic Optimization} gives the numbers when we apply simple optimizations, for example, those indicated in Section \ref{sec:translate}. Specifically, unnecessary \prog{length} computations, computations of identity sublists, and instances of $\prog{zipWith1~(\BS{}v.v)}$ were removed. We show the numbers separately for functions that use plain ellipsis expressions and those that also use split patterns.
The naive translations for plain ellipsis notation leads on average to 2.8 additional iterations, which reduces to an average of of 0.8 when using optimizations. For 12 functions that use split patterns (non-recursively) the averages are, respectively, 4.3 and 1.2. The recursive use of split patterns, which happens in 14 functions, generally leads to a quadratic overhead.

Thus the code that results from the translation of ellipsis notation is only in about 40\% of the cases as efficient as existing, more direct implementations. In the other cases, it introduces unnecessary list traversals. 
This means that ellipsis notation should be used judiciously and not indiscriminately. 

This is not a new situation for functional programming. 
Take, for example, the simple implementation of \prog{reverse}, shown in Section \ref{sec:elliexpl}, that uses \prog{(++)} and has quadratic runtime. This version is routinely used to illustrate programming with lists 
\cite[p.\ 122]{BW88}, 
\cite[p.\ 99]{Bir98}, 
\cite[p.\ 265]{Brady2017}, 
\cite[p.\ 69]{Hud00},
\cite[p.\ 50]{Hut07},
\cite[p.\ 72]{Pau96},
\cite[p.\ 81]{Stump2016}, 
\cite[p.\ 206]{Tho11}, 
\cite[p.\ 55]{Ull98}, 
even though a linear-time definition using a two-parameter helper function is widely known and not much more complicated. 
The inefficient version is then sometimes used to illustrate the derivation of a more efficient implementation, and this is exactly what can be done with potentially inefficient function definitions that employ ellipsis notation.
Similarly, lists are frequently used as a simple representation for sets and dictionaries even if binary search trees are readily available in libraries as efficient alternatives.
Starting with a succinct and easily understandable implementation that can later be optimized for efficiency (if needed) is a well accepted approach to programming. We believe the same attitude should be adopted toward ellipsis notation.

\begin{figure}[t]
\centering
\begin{tabular}{|c|c|c||c|c|}
\cline{2-5}
\multicolumn{1}{l|}{} & 
\multicolumn{2}{c||}{\textbf{Ellipsis Notation}} &
\multicolumn{2}{c|}{\textbf{+ Split Patterns}} \\
\hline
\textbf{Additional} & \textbf{Naive} & \textbf{With Basic} & \textbf{Naive} & \textbf{With Basic}\\
\textbf{Iterations} & \textbf{Translation} & \textbf{Optimizations}& \textbf{Translation} & \textbf{Optimizations} \\
\hline
0 & \bl2 &   22 & \blz & \bl4 \\
1 & \bl3 & \bl9 & \bl2 & \bl4 \\
2 &   15 & \bl4 & \bl2 & \bl2 \\
3 & \bl8 & \bl1 & \bl2 & \bl2 \\
4 & \bl7 & \bl3 & \blz & \blz \\
5 & \bl2 & \blz & \blz & \blz \\
6 & \bl1 & \blz & \bl2 & \blz \\
7 & \bl1 & \blz & \bl4 & \blz \\
\hline
\end{tabular}
\caption{Linear overhead of translated functions. For simple ellipsis notation, the second column lists the numbers for the plain translation, and the third column gives the numbers when simple optimizations have been applied. The last two columns show the numbers for the examples that also use split patterns, excluding 14 functions that use them recursively, which incurs quadratic overhead in the worst case.}
\label{fig:overhead}
\end{figure}

\section{Related Work}
\label{sec:relaw}


Numerous studies have established that programming language syntax matters.
Stefik and Siebert provide an overview of the field \cite{Stefik2013}. They also present results of two studies on intuitiveness of programming language syntax and the effects of syntax on programmers. They report evidence that language syntax matters to novice programmers and influences the accuracy of their programs.
Their results confirm earlier studies that showed that language notation has an impact on novice programmers, for example, \cite{Denny2012,Denny2011,Kelleher2005,Wiedenbeck1997}.

While there have been numerous studies of programming language features and their effect on programmers \cite{Kaijanaho2014},\footnote{Most of the studies have investigated imperative languages (especially Java), although there have also been several studies on functional languages
\cite{Tirronen2015,Singer2018,SwitchFL,HOFselection,HOFconcepts,RainfallFP,PNE22fie}.} 
there is no general theory or guide for designing the syntax of new programming language features (or for redesigning the syntax of existing ones). One notable exception is the cognitive dimensions framework \cite{GP96,Blackwell2003a} for supporting the design of notations. Cognitive dimensions are quite helpful in showing trade-offs between different goals in designing a language.


According to \cite{HaskellHistory}, the list comprehension notation was originally proposed by John Darlington and first built into the language KRC by David Turner. The name ``list comprehension'' was introduced by Phil Wadler \cite{Wad85}.
Since then, list comprehensions have been adopted by other languages and can now be found (sometimes with modifications) in Python, Racket, C\# (as LINQ), Erlang, and Javascript (as array comprehensions).
The list comprehension notation has also been suggested as a query language for relational databases \cite{Tri91}.
Moreover, the popularity of list comprehensions has spurred further language design. For example, the GHC compiler supports extensions such as parallel list comprehensions, and the use of list comprehensions as a query language has been further supported by extensions for group-by and order-by qualifiers \cite{PW07}.


Inspirations for the design of syntax can often be taken from the notations that are used in specific domains. Adopting a well-known and well-understood notation supports its acceptance and wide-spread use by programmers. 
The ellipsis symbol (``$\ldots$'') is conveniently used to elide elements from an enumerated list of items, assuming that the reader understands the pattern underlying the list and can complete the series by themselves.
Ellipses are frequently used in mathematical settings to describe unbounded series such as $1,4,9,\ldots$, as well as bounded ones, such as A, $\ldots$, Z.
In addition to defining sets, ellipses are also used to define functions, such as $\sum_{i=1}^n i=1+\ldots+n$ or $n!=1\times\ldots\times n$.
%
%
And in the realm of programming language theory, ellipses have also been used for describing program or language extensions \cite{Pie02}.


One of the attractive features offered by ellipses is that they can replace recursive definitions by shorter iterative ones \cite{Tripledots}.
However, the existing use of ellipses in programming seems limited. 
For example, in Haskell one can use two periods to denote bounded and unbounded sequences, such as the expression \prog{[1,3..]} that generates the list of odd numbers. The notation is syntactic sugar for the functions \prog{enumFrom} and \prog{enumFromThen}, which are defined for many types and can be extended by the user.
One could use the notation to approximate the mathematical definitions of, say, the factorial function as follows.

\begin{program}
fac n = foldr1 (*) [1..n]
\end{program}
However, the need for \prog{foldr1} to insert the multiplication operator between the numbers in the list makes this definition much less attractive than the mathematical version, and it also requires an understanding of \prog{foldr1}.
Moreover, one cannot use the sequence notation in patterns, which further limits its use in defining functions.

An approach for using ellipses in LISP programs was presented by \citet{ElliLisp}. With their proposed extension, we can define the factorial function in the following way.

\begin{program}
DEFINES[FAC (SLAMBDA(X)(1 * 2 * 3 ... * X))]
\end{program}

Even though this indicates a promising approach to writing programs using ellipses, we believe that the full potential of the approach can be realized only when integrated into a language with a richer syntactic basis. 
Haskell seems to be particularly promising as a candidate, since it facilitates the definition of functions with equations.
The paper also introduces special notation built on top of ellipses. For example, \prog{MU X(L)(P)} finds the first element in the list \prog{L} that satisfies the property \prog{P} (where references to \prog{x} in \prog{P} will be substituted by the elements from \prog{L}). Here \prog{L} can be expressed using an ellipsis.
A new \prog{SLAMBDA} abstraction form also allows the use of ellipses to form a kind of list pattern. For example, the function for reversing a list can be defined as follows.

\begin{program}
DEFSEQ[REV (SLAMBDA ((X1 X2 ... XN)) (XN X<N-1> ... X1))]
\end{program}
The paper provides an algorithm (called \emph{formula extrapolation}) for inferring a common pattern from the terms given in an ellipsis sequence. But the notation itself is not formalized, and its scope and expressiveness are unclear.
It is unknown to what extent this language extension has been or is currently being used. While the approach successfully incorporates ellipses, it also suffers from some non-trivial syntactic overhead imposed by the embedding into LISP. In any case, LISP may have not been the best launching pad for this language extension because, somewhat ironically, the language, by design, doesn't care much about syntax.


Macro systems, such as the one employed by the Scheme language \cite{Kohlbecker1987,Clinger1991,Dybvig1992}, have leveraged ellipses to express repeated syntactic constructs. 
Here, an ellipsis indicates that the preceding pattern may match multiple times, or that the preceding expression may repeat on expansion.
Repetitions can thus be specified on both input and output in the form of patterns and templates.
For instance, the \prog{and} macro, which returns the first false argument or the last if all are true, can be defined as follows.

\begin{program}
(define-syntax and
    (syntax-rules ()
        [(and)  #t]
        [(and e) e]
        [(and e1 e2 ...) (if e1 (and e2 ...) #f)]))
\end{program}
However, syntactic sequences captured in this manner cannot be directly operated on like typical runtime values---the ellipsis acts as more of a temporary marker interpreted entirely at the time of expansion rather than an interactive structure.


The Rhombus language \cite{Rhombus2023} expands the interactiveness of ellipses by allowing their use in list constructions and patterns. 
Here, an ellipsis serves in part as a binding construct for the remainder of a list.
For instance, summation over a list can be defined as follows. 

\begin{program}
fun
| sum([]): 0
| sum([n, m, ...]): n + sum([m, ...])
\end{program}
Our addition of folds effectively extends the functionality of Rhombus's repetition notation by allowing sequences to be collapsed in addition to being constructed, avoiding the recursive call to \prog{sum}.
We have also taken a more dynamic approach to operating on sublists by introducing an upper bound to the ellipsis, allowing arbitrary sublists to be constructed and collapsed.

An initial proposal for a subset of ellipsis notation presented here has been explored in the baccalaureate thesis \cite{Zim24}, which defines the notation via a translation into a subset of Haskell.


As the Haskell and LISP applications indicate, ellipses are particularly attractive for dealing with list data structures, and thus it is not surprising to find an approach that shows how to use ellipses in proofs about lists \cite{ElliProof}. The authors also describe an implementation of the idea in the $\lambda$\emph{Clam} proof editor.
The limitations of that approach are discussed in \cite{ElliProofLimit}.
In a related work, the Mizar proof editor allows the definition of logical infix connectives with varying number of arguments using ellipses \cite{Mizar,MizarEditor}.


Ellipses have also been used in specification languages for spreadsheets 
\cite{EACK06jfp}
%
and in proof trees to simplify program traces for the purpose of explaining computations 
\cite{BEFG21aplas,BEF23cola}.

\section{Conclusions}
\label{sec:concl}

We have presented an extensions of the pattern and expression syntax of a functional language by ellipsis notation. We have defined the meaning of the language extension through a reduction to programs without the notation. In particular, the meaning of span and fold ellipsis expressions has been given by least general generalizations that are computed via anti-unification.

We have shown that the notation covers a wide variety of function definitions on lists, which can support the introductory teaching of functional programming via an intuitive, well-known notation and methods for generating generic traces and derivations of recursive function definitions.
This might be particularly helpful for reaching people whose first exposure to programming was through an imperative programming language, since ellipsis notation is inherently iterative and thus makes programmers with an imperative programming mindset feel at home when starting to program in a functional language.

The benefits of ellipsis notation come at the price of added runtime complexity that is caused by the translation that often generates unnecessary additional list traversals. 
Approaching programming problems with high-level specifications that can be transformed into more efficient implementations has a long tradition in functional programming, and we believe this perspective should be taken when judging ellipsis notation in terms of efficiency.


In summary, ellipsis notation offers an attractive syntactic supplement to functional languages that facilitates the succinct formulation of many list functions in a more direct and thus more accessible way than equivalent recursive definitions.
Syntax-rich languages such as Haskell offer \prog{let} as well as \prog{where} blocks for local definitions, for instance. They offer equations as well as case expressions for deconstructing data type values, and some offer list comprehensions in addition to \prog{map}, \prog{filter}, and \prog{concat} for writing functions on lists.
Programming with ellipses offers another tool for enriching the syntactic variety and expressiveness of functional languages.

\bibliography{elli,me,fp}

\section*{Appendix}

\appendix

\textsc{Proof} of Lemma \ref{lem:correct} by induction over the derivation of \gen[\varnothing]{\bexp_1}{\bexp_2}{\bexp}{\subst}.

\textsc{Equal}. 
If $\bexp_1=\bexp_2$, then we have \gen[\varnothing]{\bexp_1}{\bexp_2}{\bexp_1}{\varnothing} and $\varnothing(\bexp_1)=(\bexp_1,\bexp_1)=(\bexp_1,\bexp_2)$.

\textsc{Old} and \textsc{New}. 
If \gen[\varnothing]{\bexp_1}{\bexp_2}{x}{\subst} with $x\mapsto(\bexp_1,\bexp_2)\in\subst$, then $\subst(x)=(\bexp_1,\bexp_2)$.

\textsc{Abs}. 
Suppose \gen[\varnothing]{\abs{\xL}{\beL}}{\abs{\xR}{\beR}}{\abs{\xL}{\bexp}}{\subst} 
with \gen[\varnothing]{\beL}{[x_1/x_2]\beR}{\bexp}{\subst}.
Then we know by induction that
$\subst(\bexp)=(\beL,[x_1/x_2]\beR)$, that is, $\subst_1(\bexp)=\beL$ and $\subst_2(\bexp)=[x_1/x_2]\beR$.
We have to show that 
$\subst_1(\abs{\xL}{\bexp})=\abs{\xL}{\beL}$,
which follows because $\subst_1(\abs{\xL}{\bexp})=\abs{\xL}{\subst_1(\bexp)}=\abs{\xL}{\beL}$.
We also need to show that 
$\subst_2(\abs{\xL}{\bexp})=\abs{\xR}{\beR}$, 
which follows because $\subst_2(\abs{\xL}{\bexp})=\abs{\xL}{\subst_2(\bexp)} = \abs{\xL}{[x_1/x_2]\beR} =_\alpha \abs{\xR}{\beR}$.

\textsc{App}. 
If \gen[\varnothing]{\lft{\bar{f}}\;\beL}{\rgt{\bar{f}}\;\beR}{\bar{f}\;\bexp}{\subst''}, then we know
\gen[\varnothing]{\lft{\bar{f}}}{\rgt{\bar{f}}}{\bar{f}}{\subst'}
and \gen[\subst']{\beL}{\beR}{\bexp}{\subst''}, and by induction also
$\subst'(\bar{f})=(\bar{f}_1,\bar{f}_2)$ and $\subst''(\bar{e})=(\bar{e}_1,\bar{e}_2)$.
First, we observe that $\subst'(\bar{f})=\subst''(\bar{f})$, since the additional bindings in $\subst''$ do not affect $\bar{f}$. Therefore, we also have $\subst''(\bar{f})=(\bar{f}_1,\bar{f}_2)$.
With $\subst_1''(\bar{f}\;\bexp)=\subst_1''(\bar{f})\;\subst_1''(\bexp)=\bar{f}_1\;\bexp_1$, 
and $\subst_2''(\bar{f}\;\bexp)=\subst_2''(\bar{f})\;\subst_2''(\bexp)=\bar{f}_2\;\bexp_2$,
we can thus conclude that $\subst''(\bar{f}\;\bexp)=(\bar{f}_1\;\bexp_1,\bar{f}_2\;\bexp_2)$.
%
%
\hfill$\square$

\bigskip\noindent
\textsc{Proof} of Lemma \ref{lem:minimal}.
For the proof of the lemma we consider expressions and their substitutions as equal under $\alpha$-conversion. That is, given two pairs $(e_1,\subst_2)$ and $(e_2,\subst_2)$ such that $e_1=_\alpha e_2$ evidenced by a substitution $\sigma$ with $e_1=\sigma(e_2)$, we consider $\subst_1=_\alpha\subst_2$ iff $\subst_1=\set{(\sigma(x),e)\ |\ (x,e)\in\subst_2}$. This notion of $\alpha$-equivalence extends to set operations such as $\subst_1\subseteq\subst_2$; we omit the $\alpha$ subscript for convenience in the following.

%
%
First, we note that any proper subset of \subst\ would be insufficient to obtain a generalization of \beL\ and \beR\ because any entry $x\mapsto(\bexpa,\bexpb)$ is added to \subst\ only in rule \textsc{New}, and only when $\gamma(\bexpa)\neq\gamma(\bexpb)$. Therefore, omitting any such entry would leave a difference between \beL\ and \beR\ in the computed generalization, rendering it incorrect, contradicting Lemma \ref{lem:correct}.

Now if $\gsubst=\subst$, then the lemma is trivially true with $\subst^*=\varnothing$.

For the case $\gsubst\supset\subst$, consider any entry $x\mapsto(\bexpa,\bexpb)\in\gsubst-\subst$.
We only have to consider cases for which $x\in\gexp$ because all those bindings with $x\notin\gexp$ do not contribute to a more general expression \gexp, and the lemma becomes again true with $\subst^*=\varnothing$.
By contrast, if $x\in\gexp$, then we know that $\bexpa$ must be equal to $\bexpb$ because otherwise the difference between \beL\ and \beR\ would have been picked up by rule \textsc{New} and be part of \subst. Now we can set $\subst^*=\set{x\mapsto \bexpa\ |\ x\mapsto(\bexpa,\bexpb)\in\gsubst-\subst}$, which just has the desired property of $\subst^*(\gexp)=e$.
\hfill$\square$

\end{document}